\documentclass{article}

\usepackage{echodefs}
\usepackage{spconf,amsmath,graphicx}
\usepackage{float}
\usepackage{amssymb}
\usepackage{enumitem}
\usepackage{subcaption}
\usepackage{balance}
\usepackage{amsthm}

\usepackage{xcolor}
\newtheorem{prop}{Proposition}[section]
\usepackage[noadjust]{cite}
\usepackage{soul}
\usepackage{fancyhdr}
\usepackage[hyperfootnotes=false]{}
\usepackage{hyperref}

\newcommand{\Lx}{L_{\text{x}}}
\newcommand{\Lh}{L_{\text{h}}}
\newcommand{\Lg}{L_{\text{g}}}
\newcommand{\Ln}{L_{\text{n}}}
\newcommand{\yin}{\boldsymbol{y}_{\text{in}}}

\title{Multipath-enabled private audio with noise}
\name{Anadi Chaman$^*$, Yu-Jeh Liu$^*$, Jonah Casebeer$^\dagger$, Ivan Dokmani\'c$^*$ }
\address{Departments of $^*$Electrical and Computer Engineering and $^\dagger$Computer Science\\
	University of Illinois at Urbana-Champaign}
\begin{document}
	\ninept
	\maketitle
	\begin{abstract}
		We address the problem of privately communicating audio messages to multiple listeners in a reverberant room using a set of loudspeakers. We propose two methods based on emitting noise. In the first method, the loudspeakers emit noise signals that are appropriately filtered so that after echoing along multiple paths in the room, they sum up and descramble to yield distinct meaningful audio messages only at specific \emph{focusing spots}, while being incoherent everywhere else. In the second method, adapted from wireless communications, we project noise signals onto the nullspace of the MIMO channel matrix between the loudspeakers and listeners. Loudspeakers reproduce a sum of the projected noise signals and intended messages. Again because of echoes, the MIMO nullspace changes across different locations in the room. Thus, the listeners at focusing spots hear intended messages, while the acoustic channel of an eavesdropper at any other location is jammed. We show, using both numerical and real experiments, that with a small number of speakers and a few impulse response measurements, audio messages can indeed be communicated to a set of listeners while ensuring negligible intelligibility elsewhere.

	\end{abstract}
	\section{Introduction}
	\label{sec:intro}
	
	Consider the problem of sending audio messages to different listeners in a reverberant room, while making sure that each message can only be understood by its intended recipient. Importantly, no eavesdropper anywhere in the room should be able to understand any of the messages. 
		{\let\thefootnote\relax\footnotetext{Project webpage:  \href{https://swing-research.github.io/private-audio/}{https://swing-research.github.io/private-audio/} }}

	This problem is related to personal audio zones and sound field reproduction \cite{poletti2008an,wu2011spatial,betlehem2015personal,elliott2012robustness,cai2014sound,choi2002generation,berkhout1993acoustic,ward2001reproduction,jin2013multizone} where the goal is to reproduce different sound streams in a few predefined \emph{zones} in a room while minimizing the sound level everywhere else. In most of these approaches, however, an eavesdropper with a sensitive microphone (or a good ear) can easily understand the messages. The reason is that the loudspeakers simply reproduce linearly filtered versions of desired messages which  remain highly correlated with any residual error signal. 
	
	
	To address the problem of private audio communication, we propose two methods. As an extension of our previous work \cite{liu2018cocktails}, the first approach communicates audio messages to intended \emph{focusing spots} by emitting appropriately filtered white Gaussian noise signals from loudspeakers. The filters are constructed such that after passing through specific sets of paths and time delays, these filtered random signals sum up coherently as they arrive at the target focusing points. On the other hand, they yield incoherent signals at locations with different sets of signal propagation paths. This solution is expected to work well when a room has high spatial diversity of acoustic channels.

	
	In our second approach, the idea is to send random noise from loudspeakers in addition to message signals, such that the noise signals add up to zero only at the intended listening points, while they continue to mask the messages everywhere else. This results in the interception of clean audio messages at the focusing spots while having low intelligibility at other locations. This technique is inspired by standard methods in wireless networking on jamming eavesdroppers \cite{negi2005secret,goel2008guaranteeing}. However, to the best of our knowledge, the prior works consider fading wireless channels without explicitly considering inter-symbol interference (echoes). While this could be a fair assumption for networks like WiFi where sampling times are much larger than propagation delays of wireless signals, this is not the case in room acoustics. Hence, we adapt this jamming scheme to work with long convolutional channels. 
	
	Privacy in multizone reproduction systems was first studied in \cite{donley2016improving} where the authors also use noise to mask message signals in ``quiet'' zones to reduce intelligibility. While their method is applicable in both anechoic and reverberant conditions, the performance is degraded in the presence of echoes. On the other hand, as we elaborate later, our  methods critically rely on echoes and multipath propagation. In particular, our solutions exploit the spatial diversity of room impulse responses (RIRs) across different locations in a room and the redundant degrees of freedom in signal transmission provided by multiple loudspeakers. Unlike in multizone methods, however, we can only deliver messages to a small, fixed region of space. On the other hand, we achieve good performance using a rather small number of loudspeakers and impulse response measurements (in our experiments we use only six).
	
	The problem of jamming eavesdroppers has been studied extensively in wireless communication. The theoretical foundation was laid by Shannon \cite{shannon} and later extended by \cite{csiszar1978broadcast,wyner1975wire} who showed the feasibility of secrecy if the communication channel of an eavesdropper is degraded. The methods in \cite{negi2005secret,goel2008guaranteeing,goel2005secret} use artificial noise; \cite{barros2006secrecy} showed the possibility of secret communication as a consequence of slow wireless fading. Prior works have also looked at a related problem of eavesdropper detection \cite{mukherjee2012detecting,chaman2018ghostbuster,stagner2011practical}.

	In this paper, we empirically show that unlike traditional multi-zone sound field reproduction which is usually degraded in reverberant environments \cite{jin2015theory,betlehem2005theory}, both of our proposed approaches give excellent results in the presence of echoes since echoes enhance spatial diversity. We derive conditions needed to generate desired messages at the focusing spots, and demonstrate both numerically and through real experiments that with six speakers and the knowledge of RIRs at the intended listening points, private audio communication is effectively achievable. In addition, we compare the robustness of the two approaches to system failures and uncertainties.

	

	\section{Problem Formulation}
	\label{sec:problem_formulation}
	Consider a system with $L$ loudspeakers, each emitting an audio signal to $K$ listeners. Without loss of generality, let the desired length of the signal $\boldsymbol{y}_k$ at the $k^{\text{th}}$ listener be $N$. We also assume that the room impulse response (RIR) between the $k^{\text{th}}$ listener and the $i^{\text{th}}$ speaker is a sequence $h_{ki}$ which is $L_{\text{h}}$ long {and known a priori}.
	
	This signal received by the $k^{\text{th}}$ listener is given as a sum of convolutions:
	\begin{equation}\label{linear_sum_convolve}
	{y}_k(n)=\sum_{i=1}^{L}  ({h}_{ki}*{x}_i)(n) ,\text{ }n=0,1,...,N-1,
	\end{equation}
	where $\boldsymbol{x}_i\in \mathbb{R}^{\Lx}$ is the signal transmitted by the $i^{\text{th}}$ speaker with length $\Lx=N-\Lh+1$, {and $*$ represents linear convolution}. We define intended message vector $\yin \in \mathbb{R}^{NK}$ as a concatenation of all $\boldsymbol{y}_k\in \mathbb{R}^N$: $\boldsymbol{y}_{\text{in}}=[\boldsymbol{y}_1^\T, \boldsymbol{y}_2^\T, \ldots \boldsymbol{y}_K^\T]^\T$. Similarly, we define channel matrices $\boldsymbol{H}_k$ of size $N\times L\Lx$ as $[\boldsymbol{H}_{k1}, \boldsymbol{H}_{k2}, \ldots, \boldsymbol{H}_{kL} ]$, where each $\boldsymbol{H}_{ki}$ is a Toeplitz convolution matrix composed using  $\boldsymbol{ h}_{ki}$. Defining $\boldsymbol{H}=[\boldsymbol{H}_1^\T, \boldsymbol{H}_2^\T, \ldots, \boldsymbol{ H}_K^\T]^\T$ and $\boldsymbol{x}=[\boldsymbol{x}_1^\T,   \boldsymbol{x}_2^\T, \ldots, \boldsymbol{x}_L^\T]^\T$, \eqref{linear_sum_convolve} can be rewritten as:
	\begin{equation}\label{HX_matrix_eqn}
	\yin=\boldsymbol{Hx}.
	\end{equation}
	
	If the matrix $\boldsymbol{H}$ has full row rank, we can reconstruct any desired message signals at the $K$ listeners. A well-known solution to \eqref{HX_matrix_eqn} is given by $\boldsymbol{x}=\boldsymbol{ H}^\dagger\boldsymbol{y}_{\text{in}}$, where $\boldsymbol{ H}^\dagger$ is the pseudoinverse of $\boldsymbol{ H}$. Though this solution suffices for message reconstruction at the listeners, it does not enforce unintelligibility at other locations. We could, however, exploit the additional degrees of freedom provided by the nullspace of $\boldsymbol{H}$ to generate a suitable $\boldsymbol{x}$ that ensures signal degradation outside the target focusing spots.
	
	We note that for typical audio sampling rates, RIR lengths and message lengths, $\boldsymbol{H}$ is far too large to compute the pseudoinverse explicitly. That is why we solve all least-squares design problems in this paper by the conjugate gradient method. Since the involved matrices are all block-Toeplitz, the conjugate gradient method can be efficiently implemented using fast Fourier transforms.
	
	\section{The two approaches}
	\label{the_two_approaches}
	As per \eqref{HX_matrix_eqn}, $\boldsymbol{x}$ can be suitably chosen to ensure that the message signals outside the focusing spots remain unintelligible. In this section, we present two methods to achieve this task, each constructing $\boldsymbol{x}$ in a different way: (i) multichannel convolutional synthesis (MCCS) by noise and (ii) noise in the nullspace approach.
	
	\subsection{Multichannel convolutional synthesis by noise}
	\label{sec:cmg_description}
	Recall from \eqref{linear_sum_convolve} that the signal arriving at the $k^{\text{th}}$ listener is $\boldsymbol{y}_k=\sum_{i=1}^{L}  \boldsymbol{h}_{ki}*\boldsymbol{x}_i$. In this first approach, we constrain $\boldsymbol{x}_i$ to be a convolution of a filter $\boldsymbol{g}_i$ of length $\Lg$ with a noise signal $\boldsymbol{n}_i$ of length $\Ln$, drawn from standard normal distribution. This is equivalent to 
	\begin{equation}
	\label{x_constrain_by_n}
	\boldsymbol{x}_i=\boldsymbol{N}_i\boldsymbol{g}_i,\text{ }i=1,2,\ldots,L,
	\end{equation}
	where $\boldsymbol{N}_i$ is an $\Lx \times \Lg$ Toeplitz convolution matrix composed using the vector $\boldsymbol{n}_i$, with $\Lx=\Lg+\Ln-1$. We define $\boldsymbol{g}=[\boldsymbol{g}_1^\T, \boldsymbol{g}_2^\T, \ldots, \boldsymbol{g}_L^\T]^\T$ and a block diagonal matrix $\boldsymbol{N}$ as
	$$\boldsymbol{N}=\diag([\boldsymbol{N}_1,\boldsymbol{N}_2, \ldots ,\boldsymbol{N}_L]).$$
	Then equations in \eqref{x_constrain_by_n} can be combined for all $i \in \set{1,\ldots,L}$ to give $\boldsymbol{x}=\boldsymbol{Ng}$ and 
	\begin{equation}\label{filtering_equation}
	\yin=\boldsymbol{HNg}.
	\end{equation}
	Given $\boldsymbol{HN}$ and $\boldsymbol{y}_{\text{in}}$, $\boldsymbol{g}$ can be computed using conjugate gradient method. 
	
	This model constrains $\boldsymbol{x}$ to lie on a subspace of random vectors. To understand why, consider the signal emitted by the $i^{\text{th}}$ loudspeaker, $\boldsymbol{x}_i$, which can be written as 
	$$
	x_i(n)=\sum_{p=0}^{\Ln-1} n_i(p)g_i(n-p),\text{ }n=0,1,...,\Lx-1.
	$$

	{We can interpret $\boldsymbol{x}_i$ as a sum of randomly-scaled translates of filter $\boldsymbol{g}_i$. For all speakers, $\boldsymbol{g}_i$ are constructed such that convolutions of $\boldsymbol{x}_i$ with room impulse responses sum up to yield the desired messages only at the listeners. Thus, a specific set of RIRs $\set{\boldsymbol{h}_{ki}}$, corresponding to the intended listener--speaker pairs correctly descrambles the translates. In a room with rich spatial diversity, locations other than the intended listening points will be characterized by a different set of RIRs. We thus cannot expect the descrambling to yield the correct output, and the randomness of $\boldsymbol{n}_i$ then ensures non-intelligibility of the resulting signal.}

	\subsection{Noise in the nullspace}
	
	We adapt the second approach from the wireless communications literature.
	Concretely, $\boldsymbol{x}$ is chosen as a sum of a message-carrying vector $\boldsymbol{s}\in \mathbb{R}^{L\Lx}$ and a noise-like signal $\boldsymbol{w}\in \mathbb{R}^{L\Lx}$, i.e., $\boldsymbol{x}=\boldsymbol{s}+\boldsymbol{w}$. We construct $\boldsymbol{s}$ and $\boldsymbol{w}$ to satisfy $\boldsymbol{Hs}=\yin$ and $\boldsymbol{Hw}=\boldsymbol{0}$, so that
	\begin{equation}
	\label{X_def_additive_method_expanded}
	\yin=\boldsymbol{H}(\boldsymbol{s}+\boldsymbol{w})=\boldsymbol{H}\boldsymbol{s}.
	\end{equation}

	This is achieved by choosing $\boldsymbol{w}$ as the projection of a random noise vector on the nullspace of the channel matrix $\boldsymbol{H}$, i.e., $\boldsymbol{w}=\boldsymbol{P_{\mathcal{N}(H)}}\boldsymbol{v}$, where the entries of $\boldsymbol{v}$ are i.i.d. standard Gaussian and $\boldsymbol{P_{\mathcal{N}(H)}}$ is the projector on the null space of $\boldsymbol{H}$.
	
	As mentioned in Section \ref{sec:problem_formulation},  $\boldsymbol{H}$ is typically large, which makes the direct computation of its nullspace a prohibitively complex task. Instead, we first find the projection of $\boldsymbol{v}$ on the row space of $\mH$ by solving
	\begin{equation}\label{eq:row_space}
	\hat{\boldsymbol{z}}=\underset{\boldsymbol{z}}{\mathrm{argmin}} \ \| \boldsymbol{ v}-\boldsymbol{H}^\T\boldsymbol{z} \|_2^2.
	\end{equation}
	We again use the conjugate gradient method to solve  \eqref{eq:row_space} using fast Fourier transforms since $\mH$ is block-Toeplitz. Once $\hat{\vz}$ is found, the nullspace projection $\boldsymbol{P_{\mathcal{N}(H)}}\boldsymbol{v}$ is simply $\boldsymbol{v}-\boldsymbol{H}^\T\boldsymbol{\hat{z}}$.

	\begin{figure*}
		\centering
		\includegraphics[width=1\textwidth]{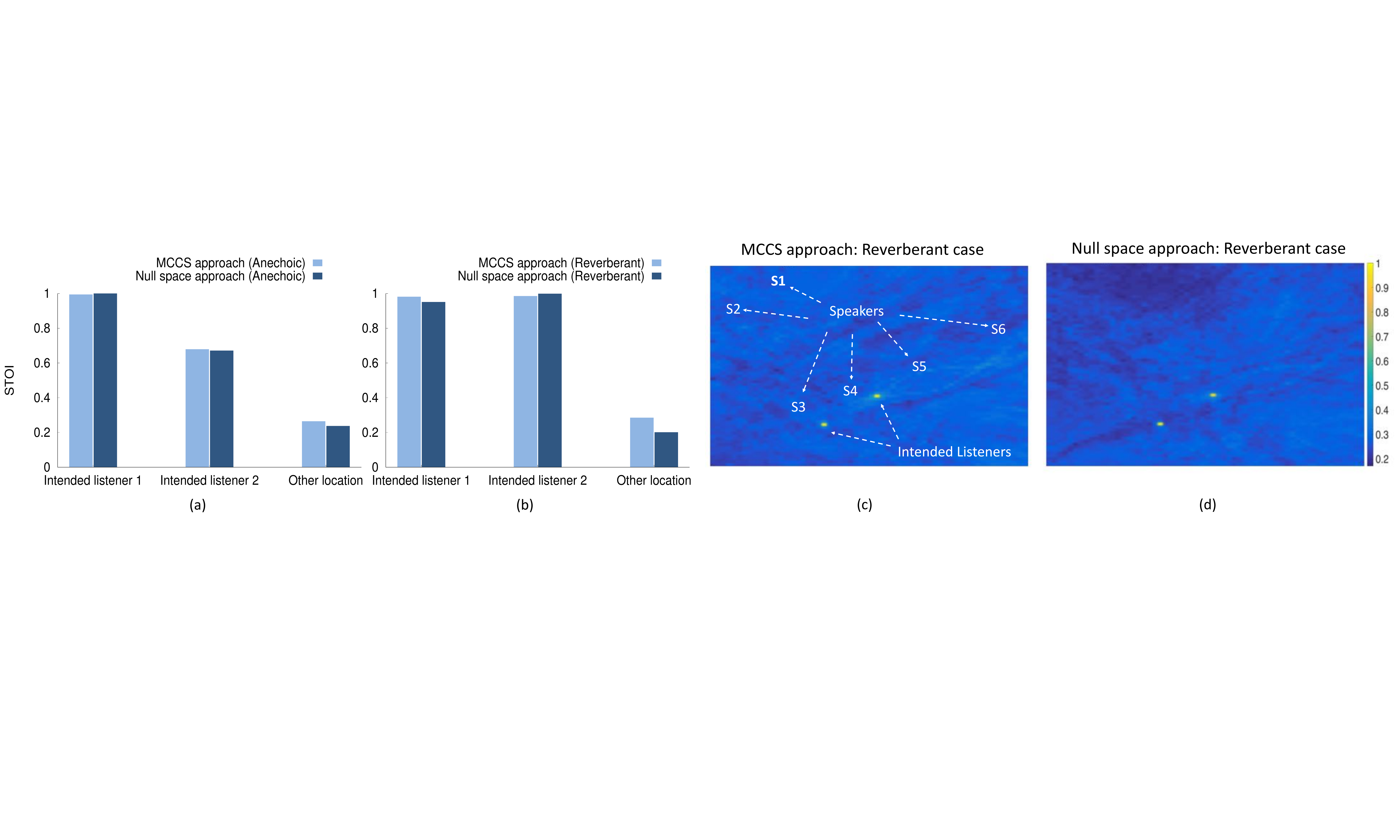}
		\caption{STOI scores at 2 intended listeners and one additional location using MCCS and nullspace approach in (a) anechoic and (b) reverberant setting.  (c)-(d) Heat maps reflecting STOI scores at 4200 locations in a simulated room of size $7\text{ }$m$\times 8\text{ m}$. Speakers illustrated as S1-S6.}
		\label{fig:stoi_heat_map}
	\end{figure*}

	\section{Conditions for perfect reconstruction}
	\label{perfect_recon}
	In this section, we present the conditions needed to ensure perfect reconstruction of any set of message signals of length $N$ at the $K$ listeners (or any $\yin \in \mathbb{R}^{NK}$) for both approaches. 
	
	\subsection{Multi-channel convolutional synthesis by noise}
	\label{sec:perfect_recon_filtered}
	
	From \eqref{filtering_equation}, perfect reconstruction can be achieved if the overall channel matrix $\boldsymbol{HN}$ has full row rank, $NK$. 
	
	{We make the assumption that the room is drawn randomly from a continuous distribution. (For example, let the corners be chosen uniformly at random within fixed balls.) We also assume that the loudspeaker and listener positions are placed at random according to an absolutely continuous distribution. These assumptions imply that the distribution of the nullspace of $\mH$ is absolutely continuous with respect to the Haar measure on the Grassmannian.} Then, we have the following result.
	\begin{prop}
		\label{prop4}
		Suppose $L L_g \geq NK$. Then $\mH \mN$ has full row rank with probability $1$.
	\end{prop}
	
	\begin{proof}
		We have that $\rank(\mH \mN) \leq \min \set{ \rank(\mH), \rank(\mN)}$ by rank inequalities. With the conditions of the proposition, this implies that $\rank(\mH \mN) \leq NK$. The only way to have a strict inequality is that the nullspace of $\mH$ intersects the range of $\mN$ along a subspace of dimension greater that $LL_g - NK$. On the other hand, because the nullspace of $\mH$ is continuously distributed and independent from $\mN$, it will intersect the range of $\mN$ exactly along a subspace of dimension $L L_g - NK$ with probability 1.
	\end{proof}

	This result implies that for most setups in sufficiently reverberant rooms, we will be able to produce the desired messages at the listener positions.

	\subsection{Noise in nullspace approach}
	\label{sec:recon_additive_noise}
	From \eqref{X_def_additive_method_expanded}, $\boldsymbol{ H}$ needs to have full row rank for perfect reconstruction of all $\yin \in \mathbb{R}^{NK}$. Similar to the previous case, since $\boldsymbol{ H}$ is a function of the RIRs between the speaker-listener pairs, it is not completely in the user's control to ensure that it has full rank as it depends on room geometry and the spatial diversity of RIRs. {In practice, however, if we assume a randomized setup and room as in the previous section, and the conditions of Proposition \ref{prop4_null} are satisfied, then $\boldsymbol{H}$ can be expected to have full row rank with probability 1.}
	\begin{prop}
		\label{prop4_null}
		{The following conditions are necessary for perfect reconstruction of message signals at the listeners.}
		
		\begin{enumerate}[label=(\alph*)]
			\item The number of rows of $\boldsymbol{ H}$ should be at least as large as the length of $\yin \implies (\Lx+\Lh-1) \geq N$. 
			\item There should be at least as many columns as rows in $\boldsymbol{ H}$. 
			\item $\Lx$ needs to be greater than the highest relative time delay among each listener-speaker pair. 
			
		\end{enumerate}
		
	\end{prop}
	\begin{proof}
		$(a)$ ensures that we have sufficient samples to generate the desired message length; $(b)$ is elementary linear algebra; $(c)$ ensures that ``silent'' regions do not exist within a signal generated at a listening point.
	\end{proof}

	{It should be noted that both of our approaches satisfy the condition in (a) with equality. Also, $(b)$ gives a lower bound on the number of speakers, $L$, needed for reconstruction, i.e., $L\geq \frac{NK}{\Lx}$. This is lower than the number of speakers needed by the MCCS approach, as per Proposition \ref{prop4} 
	}
	
	\section{Experimental Results}
	\label{sec:Results}
	We evaluate the performance of the two proposed techniques using both numerical and real experiments. The numerical experiments are performed with 6 loudspeakers randomly placed in a simulated convex room of size $7$~m~$\times~8$~m having walls with absorption coefficient 0.35. RIRs between the speakers and listeners are calculated based on image source model, using the \texttt{pyroomacoustics} package \cite{scheibler2018pyroomacoustics}. We perform the real experiments in an office space of size $10$ m $~\times~6$ m using two Genelec 8030B and
	four Genelec 8010A loudspeakers. The RIRs are measured using the exponential sine sweep technique \cite{farina2000simultaneous}. In all experiments, the power of signals emitted by the loudspeakers is kept fixed. The intelligibility of the generated sounds is assessed using Short-Time Objective Intelligibility (STOI) \cite{taal2010short} measure.

	\subsection{Numerical experiments}
	\subsubsection{Perfect reconstruction: A case for echoes}
	In order to provide insight into the importance of echoes in our solution, we first perform an experiment in a simulated anechoic room.
	We randomly place two listeners inside the room and calculate STOI scores of the signals arriving there using the two approaches. An additional location is randomly chosen to examine the signal degradation outside the target focusing spots. We then repeat the same experiment but in the presence of echoes. Fig. \ref{fig:stoi_heat_map} (a) shows that in the anechoic setting, while the signal at the first listener has high intelligibility with STOI scores close to 1 for both approaches, the second listener does not. On the other hand, Fig. \ref{fig:stoi_heat_map}(b) shows that in the presence of echoes, signal intelligibility is restored at the second listener as well. This indicates that the spatial diversity provided by echoes helps in conditioning the channel matrix $\mH$, which in turn supports perfect reconstruction of messages at target locations.
	\subsubsection{Signal degradation outside focusing spots}
	Both Fig. \ref{fig:stoi_heat_map} (a) and (b) indicate that the nullspace-based method has a greater impact on signal degradation at the location chosen outside the focusing spots. To examine this further, we calculate STOI scores at 4200 locations in a simulated reverberant room and create heat maps as shown in Fig. \ref{fig:stoi_heat_map} (c) and (d). {In both plots, the bright spots at the locations of intended listeners indicate high intelligibility. However, regions outside the focusing spots in Fig. \ref{fig:stoi_heat_map} (d) have relatively lower STOI scores as compared to Fig. \ref{fig:stoi_heat_map} (c), thus indicating towards better jamming capabilities of the nullspace approach.}


	
	Both methods perform signal degradation outside the focusing spots using noise. To understand how these random signals result in unintelligibility of sound, we first investigate the role of noise variance. For 100 randomly selected speaker-listener configurations, we check the impact of increasing noise variance on STOI values for both methods. Fig. \ref{fig:degradation_plot} (a) shows a decline in median STOI scores as the input noise power is increased for the nullspace approach, whereas they do not change much for the MCCS method. 
	
	\begin{figure}[t]
		{\includegraphics[width=1.03\linewidth,height=4.2cm]{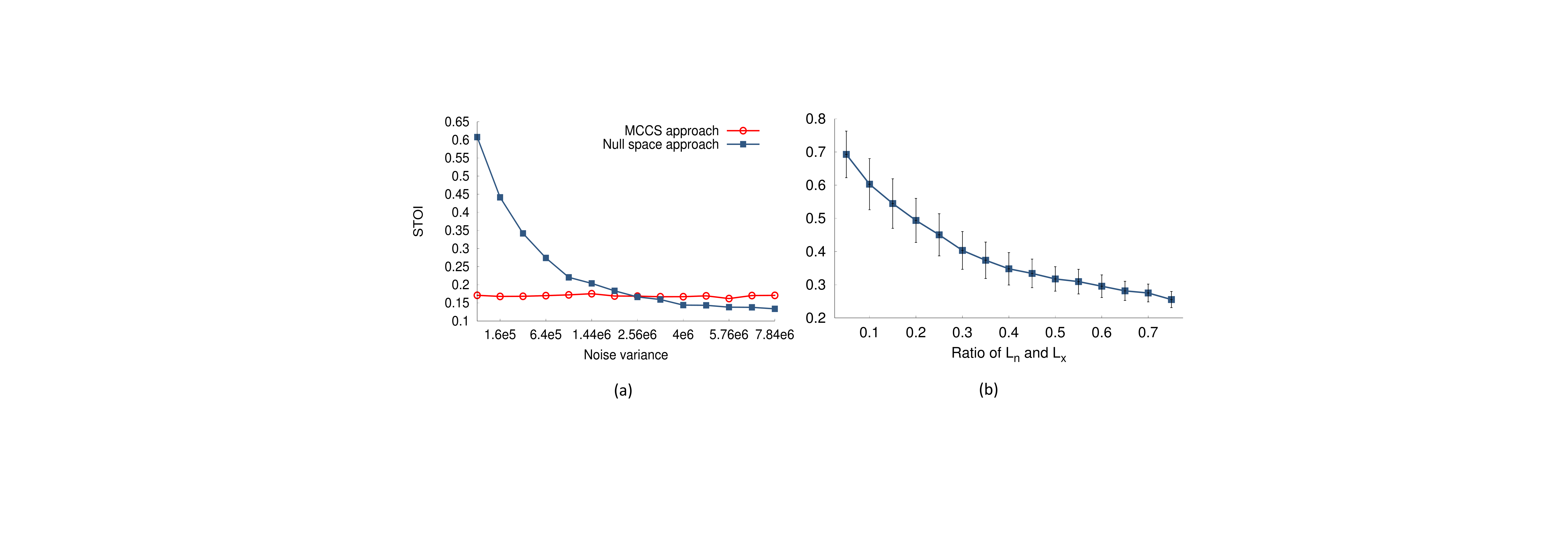}}
\caption{(a) STOI vs noise variance for the 2 methods outside focusing spots. (b) STOI vs noise length as a proportion of overall input length for MCCS approach outside focusing spots.}
		\label{fig:degradation_plot}
	\end{figure}

	This result is not surprising because in the nullspace approach, noise is fed into the loudspeakers with the message signals in an additive sense. Thus, a deterioration of SNR and subsequent STOI decline is expected with increase in noise variance. However, the signal emitted by the $i^{\text{th}}$ loudspeaker is $\boldsymbol{x}_i=\boldsymbol{n}_i*\boldsymbol{g}_i$ for MCCS method. Here, if the variance of $\boldsymbol{n}_i$ is increased, $\boldsymbol{g}_i$ simply gets scaled to preserve the original $\boldsymbol{x}_i$.

	We now investigate the factors that impact the jamming capability of the MCCS approach. Recall that this method involves ``scrambling'' of message-carrying input filters $\boldsymbol{g}_i$ by noise which are thereby appropriately descrambled at the intended locations by the correct RIR values. Thus, we expect that longer noise vectors would have a stronger impact on signal integrity when the RIR changes. To verify this claim, we vary the length of noise vectors $\Ln$ as a proportion of a fixed length $\Lx$, and calculate the STOI scores for 100 randomly chosen speaker-listener configurations. Fig. \ref{fig:degradation_plot}(b) verifies that increasing the length of noise vectors leads to a decrease in median intelligibility scores outside the focusing spots.
	
	
	{
		These results point towards an interesting phenomenon. Given unlimited available input power at the speakers, one could arbitrarily improve jamming by increasing noise power in the nullspace method. However, in MCCS approach, an arbitrary increase in jamming by increasing $\Ln$ is not feasible, because for a fixed message length $N$ and fixed $\Lh$, $\Lx=\Lg+\Ln-1$ is fixed, and one can only increase $\Ln$, as long as $\Lg\geq \frac{NK}{L}$ (from Proposition \ref{prop4}).
	}

	\vspace{-2mm}
	
	{
		\subsubsection{Robustness to system failures and uncertainties}
		\label{sec: robustness}
		We assess how the reconstruction of audio messages at the target listeners is affected by system failures and uncertainties: (i) malfunction of loudspeakers while emitting audio signals, and (ii) errors in RIR measurements. We did simulations over 100 random speaker--listener configurations and examined the behavior of the STOI scores. 
		In (i), we compute the appropriate $\mathbf{x}_i$ (to be emitted by the $i^{\text{th}}$ loudspeaker) for a system of 6 speakers. However, while measuring STOI at the listeners, not all speakers are used. Fig. \ref{fig:robustness_image}(a) shows that the STOI scores decline as more speakers are dropped, and the decline is more rapid for the nullspace method as compared to MCCS approach. 
		
		On the other hand, we analyze the robustness to channel measurement errors by computing $\mathbf{x}_i$ using RIR values with white Gaussian noise added to them. These erroneous $\mathbf{x}_i$ are then convolved with the true RIRs to compute the signals arriving at focusing spots. Fig. \ref{fig:robustness_image}(b) indicates that errors in the knowledge of RIRs before signal transmission by the loudspeakers lead to reduced intelligibility at the focusing spots. Again, the MCCS approach shows more robustness to uncertainties as compared to the nullspace approach.

	}

	\begin{figure}[t]
		{\includegraphics[width=1.03\linewidth,height=4cm]{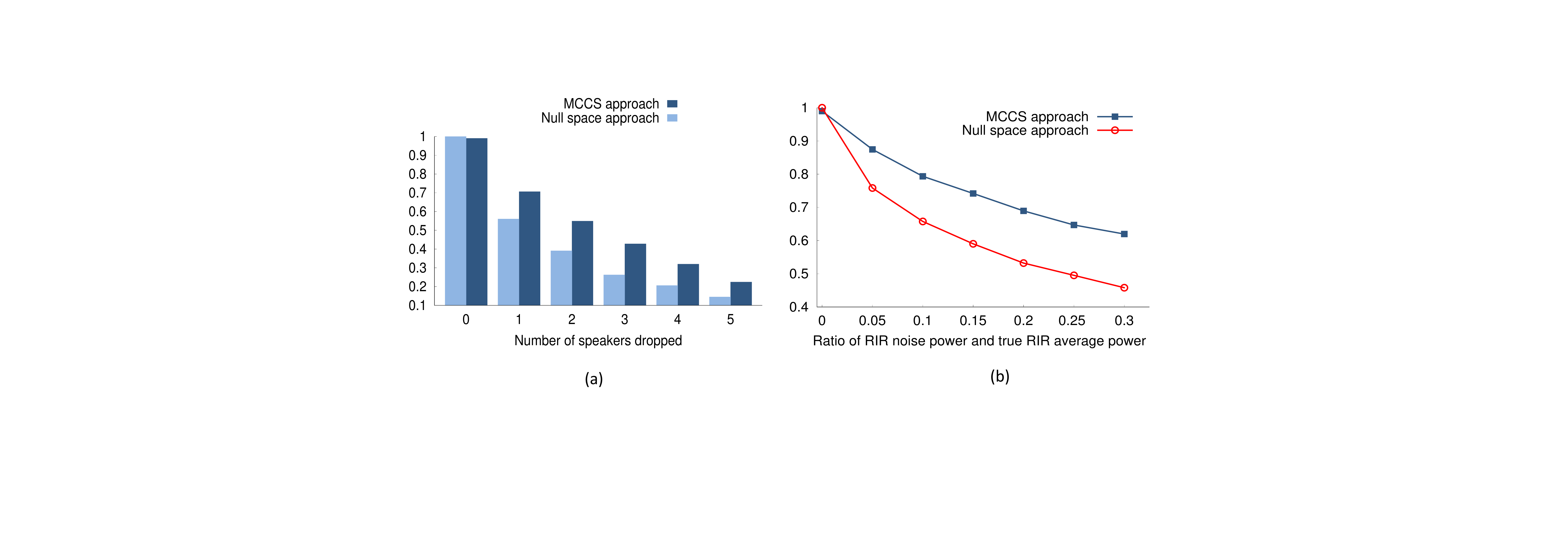}}
		\caption{Robustness analysis. Impact of (a) speaker malfunction and (b) inaccuracies in RIR estimates on STOI scores at focusing spots. }
		\label{fig:robustness_image}
	\end{figure}

\vspace{-1mm}
	
	\subsection{Experiment in a real setting}
	\label{sec: real_Experiments}
	We perform an experiment to evaluate the two approaches in a real room with 6 loudspeakers and measure the STOI scores of generated sounds with microphones at 7 locations. The experimental setup is shown in Fig. \ref{fig:stoi_real_low_noise_real_Exp} (a). Two microphones are chosen to be the focusing spots, and the rest are placed at increasing distances from Spot 2. Fig. \ref{fig:stoi_real_low_noise_real_Exp} (b) shows the measured STOI values. The observed intelligibility at the two spots is good with high STOI scores, and the signals become considerably degraded 50 cm away from the focusing spots. As expected from simulations, the nullspace approach has a stronger impact on signal degradation outside the target listeners.
	
	\begin{figure}
		
		\begin{minipage}[b]{1.0\linewidth}
			\centering
			\centerline{\includegraphics[width=1.01\linewidth,height=3.6cm]{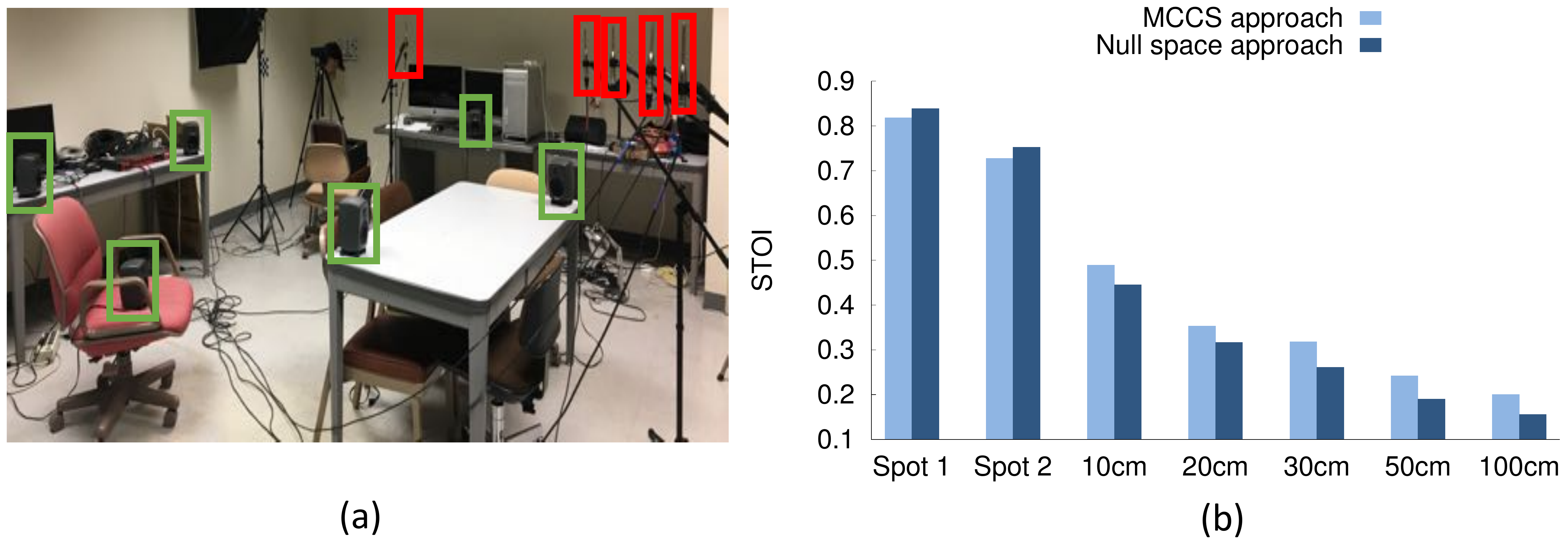}}
		\end{minipage}
		
		\caption{(a) Experimental setup: speakers represented in green, and microphones in red boxes. (b) STOI values measured at two focusing spots, and at different distances from Spot 2 in a real room setting. }
		\label{fig:stoi_real_low_noise_real_Exp}

	\end{figure}

\vspace{-1mm}
	
	\section{Conclusion}
	\label{sec:conclusions}

We present two approaches to address the private audio communication problem in a reverberant room. Both approaches are based on emitting noise signals from loudspeakers and then utilizing echoes in the room to ensure that they yield intelligible messages at selected locations, while being incoherent elsewhere. Simulated and real experiments suggest that with just 6 loudspeakers and a few impulse response measurements, we can deliver clear audio messages at the desired locations while ensuring unintelligibility everywhere else. 
	\bibliographystyle{IEEEtran}
	\balance
	\bibliography{IEEEabrv,ourbib1}

\end{document}